\documentclass[submission,copyright,creativecommons]{eptcs}
\providecommand{\event}{PLACES 2013} % Name of the event you are submitting to

\usepackage[utf8x]{inputenc}
\usepackage{amsmath}
\usepackage{amssymb}  
\usepackage{amsthm}     
\usepackage{semantic}
\usepackage{alltt}
\usepackage{url}
\usepackage{multirow}
\usepackage{esvect} 
\usepackage{eqnarray}
\usepackage{longtable}
\urldef{\mail}\path|{rde, wva}@info.fundp.ac.be|

\newcommand{\codeinmath}[1]{\textrm{{\tt #1}}} 
\newcommand{\code}[1]{\textrm{{\tt #1}}}
\newcommand{\caseof}[2]{\code{case} \ #1 \ \code{of} \ #2 }
\def\evalcont{\mathbb{C}}

\newcommand{\readTVarONLY}{\ensuremath{\code{readTVar}}}
\newcommand{\readTVar}[1]{\ensuremath{\readTVarONLY \ #1}}
\newcommand{\writeTVarONLY}{\ensuremath{\code{writeTVar}}}
\newcommand{\writeTVar}[2]{\ensuremath{\writeTVarONLY \ #1 \ #2}}
\newcommand{\bind}{\ensuremath{\code{>>=}}}
\newcommand{\bindTM}{\ \ensuremath{\code{>>=}} \ }

\newcommand{\returnTMONLY}{\ensuremath{\code{return}}}

\newcommand{\returnTM}[1]{\ensuremath{\returnTMONLY \ #1}}

\newcommand{\pt}{\ensuremath{\mathcal{T}}}
\def\evalcont{\mathbb{C}}

\newtheorem{definition}{Definition}

\newtheorem{theorem}{Theorem}
\newtheorem{corollary}{Corollary}

\title{Static Application-Level Race Detection in STM Haskell using Contracts}
\author{Romain Demeyer
\institute{University of Namur, Belgium}
\email{romain.demeyer@unamur.be} 
\and
Wim Vanhoof
\institute{University of Namur, Belgium}
\email{wim.vanhoof@unamur.be}}
\def\titlerunning{Static Application-Level Race Detection in STM Haskell using Contracts}
\def\authorrunning{R. Demeyer \& W. Vanhoof}

\begin{document}
\maketitle

\begin{abstract} 
Writing concurrent programs is a hard task, even when using high-level synchronization primitives such as transactional memories together with a functional language with well-controlled side-effects such as Haskell, because the interferences generated by the processes to each other can occur at different levels and in a very subtle way. 
The problem occurs when a thread leaves or exposes the shared data in an inconsistent state with respect to the application logic or the real meaning of the data.  
In this paper, we propose to associate {\em contracts} to transactions and we define a program transformation that makes it possible to extend \emph{static contract checking} in the context of STM Haskell.  
As a result, we are able to check statically that each transaction of a STM Haskell program handles the shared data in a such way that a given consistency property, expressed in the form of a user-defined boolean function, is preserved. 
This ensures that bad interference will not occur during the execution of the concurrent program. 
\end{abstract}

\section{Introduction}
\label{introduction}

\emph{Software Transactional Memory} (STM) \cite{software-transactional-memory} is supposed to help us in the complex task of writing concurrent programs. 
The pure and lazy functional language Haskell %TODO references
 proposes a particularly clean and straightforward integration of STM \cite{beautiful-concurrency,transactional-memory-data,composable-memory-transactions} in its concurrent version. 
Shared variables, called \emph{transactional variables} (\emph{TVar}s) in the context of STM, can be accessed by different threads using the STM primitives \texttt{readTVar} and \texttt{writeTVar}, and the programmer can protect those accesses from the interference of other threads by making them (conceptually) atomic using the primitive \texttt{atomically}.
In fact, \emph{TVar}s can \emph{only} be accessed from within such atomic blocks, also called \emph{transactions}.
While the use of STM Haskell  allows to avoid many tricky low-level bugs, such as low-level race conditions and deadlocks, this in itself is not an absolute guarantee for correctness \cite{beautiful-concurrency}. %TODO rajouter des references
Indeed, in spite of STM being a beautiful tool that allows one to get rid of low-level locking mechanisms and to focus on higher-level aspects of the program, it does not prevent all errors related to concurrent programming. In particular, a fundamental difficulty related to concurrent programming with shared variables remains: the risk of exposing those data in an inconsistent state \emph{with respect to the application logic} \cite{verifying-correct-usage}.

Let us illustrate this problem on a simple example. Consider the source code implementing a function \texttt{f} presented in the top part of Fig.~\ref{exempleintro}, where the {\tt do} notation refers to the classic syntactic sugar to express monadic computations \cite{transactional-memory-data,report-programming-language}.   
We suppose the existence of two \emph{TVar}s: \texttt{shSum}, which stores an integer, and \texttt{shTab} which stores a list of integers.
There are two nested functions in the main function {\tt f}.
The first one, \texttt{addTab}, consists of two STM operations that allows to update \texttt{shTab} by adding an integer {\tt n} in front of the list\footnote{Let \texttt{x} be an integer and \texttt{xs} be a list of integers $\mathtt{[x_1,...,x_n]}$, \texttt{(x:xs)} correspond to the list $\mathtt{[x,x_1,...,x_n]}$.}. The second one, \texttt{addSum}, updates the other \emph{TVar}, \texttt{shSum}, by adding {\tt n} to its value. 
When considering this definition, each update is protected by the operation {\tt atomically}, and, consequently, this code is free of low-level race conditions (i.e. concurrent access to the same data element with at least one access being a write). 
This is both sufficient and efficient, as long as the values of the \emph{TVar}s are independent. 
However, if there is an implicit link between the values of those variables, the story is more subtle. Suppose that {\tt shSum} is meant to represent at all times the sum of the integers in the list {\tt shTab}.
In this case, an inconsistent state (in which the value of {\tt shTab} has been updated while the value of {\tt shSum} has not) is exposed {\em between} the two transactions, which may be problematic in a concurrent program. 
Indeed, suppose for example that the thread's execution is (conceptually) suspended at this precise point, while another thread doubles the value of {\tt shSum} and that of each integer of {\tt shTab}. 
At the end of the execution of both threads, the sum of the integers in the list {\tt shTab} and the value of {\tt shSum} will be different, breaking the programmer's intention and thus showing an unacceptable error in the program. 
This situation is what is sometimes called an {\em application-level race condition} \cite{verifying-correct-usage,rdemeyer:framework-verifying-application}, as it represents an inconsistency with respect to the {\em logic} of the application that cannot be observed from the source code alone. 
In the context of our example, it can be easily corrected by encapsulating both updates in a single transaction, as depicted by the alternative implementation of \texttt{f} sketched at the bottom part of Fig.~\ref{exempleintro}.

\begin{figure}[htb] 
\centering
%\begin{tabular}{p{0.5\textwidth}|p{0.5\textwidth}}
%  \begin{minipage}{0.5\linewidth}
    \begin{footnotesize}
\begin{alltt}
f n = let addTab n = do tab <- readTVar shTab
                        writeTVar shTab (n:tab)
          addSum n = do s <- readTVar shSum
                        writeTVar shSum (s+n)
          in do atomically ( addTab n )     
                atomically ( addSum n )     
\end{alltt}
    \end{footnotesize}
\begin{center}
\line(1,0){100}
\end{center}
    \begin{footnotesize}
\begin{alltt}
f n = let ...
          in do atomically ( do addTab n    
                                addSum n )  
\end{alltt}
    \end{footnotesize}
%  \end{minipage}
% & 
% \begin{minipage}{0.5\linewidth}
% \end{minipage}

% \\
%\end{tabular}
\caption{The function \texttt{f} implemented with two transactions (top part) and with a single transaction (bottom part)}
\label{exempleintro}
\end{figure}
 
As application-level race conditions appear often and in a more subtle form in large programs \cite{learning-mistakes-comprehensive} and as they are hard to prevent with testing, there is a certain interest in having a tool that is able to detect them statically. This boils down to verifying that each transaction preserves the \emph{TVar}s in a state that is consistent with respect to the given application logic. 

One interesting approach towards specification and verification of program properties of Haskell programs is so-called \emph{static contract checking} \cite{static-contract-checking} which has been developed for a core version of the language. Its convenience lies in the fact that the property to be checked can be specified by writing it in the form of a Haskell function, which liberates the programmer from the need of dealing with a separate specification language \cite{static-contract-checking}. Unfortunately, not being designed to handle concurrent programs, the technique does not handle mutable states nor transactions. The goal of this paper is to overcome this limitation.

More specifically, we make the following contributions:
\begin{enumerate}
\item We define contract checking for the language used in the transactions of STM Haskell programs. For this purpose, we have defined a novel kind of contract, dedicated to STM operations. 
%We define a novel kind of contract, dedicated to STM operations, and we define contract checking for \emph{any} expression that can occurs in the transactions of a STM Haskell program
\item We re-express the problem of the detection of application-level race conditions in the context of contract checking.
\item We propose a practical sound method to prove automatically contract satisfaction.
\end{enumerate}
The method we propose in order to achieve that last goal is to transform expressions and contracts in a such form that an existing verification technique, introduced in Section~\ref{background}, can be used.
Our framework is presented in broad terms with the help of a motivating example in Section~\ref{idea}. Then, we define formally the framework, we prove our transformation to be correct and we discuss extensions and limitations in Section~\ref{details},  
discuss how we could overcome some current limitations of our approach, 
before replacing our results in the context of related work (Section~\ref{conclusion}).

\section{Background: Static Contract Checking for Haskell}
\label{background}

The static verification framework of Xu et al.\cite{static-contract-checking} is based on the notion of \emph{contracts} \cite{eiffel-language,spec-programming-system}. 
A contract can be seen as a refinement of the type of a function. 
For example, let us consider the function \texttt{inc} depicted in Fig~\ref{fig:inc}. 
The type of the function tells us that it takes an integer as argument and returns an integer as well.
The contract gives more information about the function by telling us that the integer expected as argument has to be strictly positive,
and that the value returned has to be strictly greater than the argument. 

\begin{figure}[htb]
  \centering
  \begin{footnotesize}
\begin{verbatim}
inc :: Int -> Int                      -- Type
inc :: { x | x > 0 } -> { r | r > x }  -- Contract
inc x = x + 1                          -- Definition
\end{verbatim}
  \end{footnotesize}
  \caption{The \texttt{inc} function, with its type and contract.}
  \label{fig:inc}
\end{figure}

In the context of this example, contract \emph{checking} consists then in verifying that if the argument fulfills its part of the contract, then the value returned fulfills its own. In the framework of~\cite{static-contract-checking}, this checking is done in two steps. First, the contracts are integrated into the function definition in a such way that the function explicitly fails by returning a special value if they are not fulfilled. 
The function transformed in this way is called the \emph{wrapped function}.
Secondly, symbolic execution is used to check whether the wrapped function \emph{can} effectively fail. For our example, the wrapped function is depicted in Fig~\ref{fig:inc:step1}. 
The outer case expression represents the fact that we assume a strictly positive value for the argument, i.e. the opposite branch is explicitly tagged as unreachable (by returning the special value \texttt{UNR}). The inner case expression represents the fact that the function must fail (by returning the special value \texttt{BAD}) if the returned value is not greater than the input. In other words, \texttt{inc'} behaves just like \texttt{inc}  except that it returns \texttt{UNR} if its argument is negative and \texttt{BAD} if the function definition violates the contract, i.e. returns a value that is no strictly greater than the input.

The second step consists in simplifying the wrapped function in order to prove that all \texttt{BAD} branches can be safely removed, 
which is quite easy in this example with simple symbolic execution and basic theorem proving which can replace \texttt{x + 1 > x} by \texttt{True} \cite{static-contract-checking}.
Note that when the function is \emph{called}  
in the program, the call in question is also replaced by a wrapped call, similar to the wrapped function definition apart from the fact that \texttt{BAD} and \texttt{UNR} are swapped. 
This depicts the fact that checking a function definition corresponds to verifying the postcondition, assuming that the precondition holds, while checking a function call corresponds to verifying that the precondition holds and then, if so, assuming that the postconditions holds.
This method allows to verify entire programs in a modular way and to deal adequately with recursion \cite{static-contract-checking}.

\begin{figure}[htb]
  \centering
  \begin{footnotesize}
\begin{alltt}
inc' x = case x > 0 of 
           True -> case x + 1 > x of
                     True  -> x+1
                     False -> BAD
           False -> UNR 
\end{alltt}
  \end{footnotesize}
  \caption{The function that is build based on \texttt{inc} and its contracts.}
  \label{fig:inc:step1}
\end{figure}

Contract checking of a function being a undecidable problem, it has three possible outcomes: either the function is definitely safe (all \texttt{BAD}s are removed during simplification and hence proven to be unreachable), either definitely not safe (the expression does simplify to \texttt{BAD}), or unknown (some \texttt{BAD}s remain present after simplification, but we cannot prove that the expression will actually fail). However, by using a suitable inlining/simplification
strategy, a considerable amount of programs can be proven to be correct with respect to their given contracts~\cite{static-contract-checking}. As an example, our own prototype implementation based on~\cite{static-contract-checking} succeeds in the verification of the somewhat more involved example represented in Fig.~\ref{fig:addpure}.

\begin{figure}[htb]
  \centering
  \begin{footnotesize}
\begin{verbatim}
add :: Int -> ([Int],Int) -> ([Int],Int)
add :: { x | True} -> { (tab,s) | sum tab == s} -> { (tab,s) | sum tab == s} 
add n (tab,s) = (n:tab,s+n)

sum :: [Int] -> Int
sum xs = case xs of []     -> 0
                    (l:ls) -> l + sum ls
\end{verbatim}
  \end{footnotesize}
  \caption{The \texttt{add} function and its type and contract.}
  \label{fig:addpure}
\end{figure}

\section{The Main Idea}
\label{idea}

In this section, we present in an intuitive fashion 
a framework
that allows to statically detect application-level race conditions in a STM Haskell program. It consists, in other words, in verifying that each transaction preserves the transactional variables (\emph{TVar}s) in a consistent state.
The main idea behind the verification is to deduce contracts from the properties expressing consistency of the \emph{TVar}s, and to subsequently verify these contracts by an adaptation of the static contract checking framework in order to make it deal with STM Haskell.
The basic difficulty in using the framework of Xu et al.~\cite{static-contract-checking} is that the latter is not designed to deal with \emph{mutable} variables and side-effects, i.e. STM and I/O primitives. However, since transactions in STM Haskell do not produce side effects other than updating the values of some \emph{TVar}s, they can be seen as pure functions taking the values of the \emph{TVar}s they manipulate as input, and producing a set of new values for them. 
To illustrate our approach, we will show how it is capable of detecting an inconsistency in the function of the top part of Fig.~\ref{exempleintro}, while it proves the alternative function (sketched at the bottom part) to be application-level race condition free.

As a first step, one needs to \emph{specify} what it means for the \emph{TVar}s to be in a consistent state. This can be done by writing a Haskell function 
that returns \texttt{True}, respectively \texttt{False}, if the set of \emph{TVar}s are in a consistent, resp. inconsistent, state.  
In the context of our running example, this function 
would be as follows:

\begin{footnotesize}
\begin{verbatim}
inv (shTab,shSum) = sum shTab == shSum
  where sum xs = case xs of []     -> 0
                            (l:ls) -> l + sum ls
\end{verbatim}
\end{footnotesize}
Indeed, in our program, the \emph{TVar}s are considered to be in a consistent state if the sum of the elements from the list stored in \texttt{shTab} equals the integer stored in \texttt{shSum}.  
Note that the above expression is the \textit{only} information that needs to be specified for our approach to be capable of verifying the absence of violations of this consistency definition. Moreover, our framework allows this function to be \textit{any} Haskell function that returns a boolean, including functions whose definition involves calls to recursive functions that are defined elsewhere in the program. 

From the above function, we generate the following \emph{STM contract}, that we call the \emph{transactional invariant}:

\begin{footnotesize}
\begin{verbatim}
INVARIANT :: || c <> c || Any     where c = {(shTab,shSum) | inv(shTab,shSum)} 
\end{verbatim} 
\end{footnotesize}
While the language of the contracts will be formally defined further down the paper, intuitively the above contract specifies that if the \emph{TVar}s are in a consistent state at the very beginning of the transaction -- i.e. their content satisfies the contract \texttt{c} at the left of \texttt{<>}, then, it must also be \texttt{True} at the very end of the transaction -- i.e. their content satisfies the contract \texttt{c} at the right of \texttt{<>}. Formally, like any Haskell expression, also a transaction in STM Haskell returns a value, but in the example we don't care about it -- hence the \texttt{Any} in the contract. 

To achieve the verification of this contract, we define an operator, which we denote by  $\pt$, that transforms a STM Haskell expression $e$, i.e. an expression which involves mutable variables and STM primitives, into a basic non-concurrent Haskell  expression $\pt(e)$ -- i.e. an expression which is completely \emph{pure} -- in such a way that the contracts can be checked on $\pt(e)$ by the non-concurrent framework of \cite{static-contract-checking}, while the results of the analysis are valid for the contracts in the original concurrent program $e$.
The intuitive idea behind the transformation is to represent the effect of a transaction on the \emph{TVar}s by a pure function (a lambda abstraction) that takes as arguments not only the potential free variables of the transaction, but also the values of those \emph{TVar}s as input, and that computes a vector containing the value computed by the transaction \emph{and} the values of the (updated) \emph{TVar}s. 

\newcommand\bslash{\symbol{`\\}}

In our example, we transform the three transactions from Fig.~\ref{exempleintro} into the three following lambda expressions:

\begin{footnotesize}
\begin{alltt}
{\bslash}n (shTab,shSum) -> ((),(n:shTab,shSum))    
{\bslash}n (shTab,shSum) -> ((),(shTab,shSum+n))    
\end{alltt}
\end{footnotesize}
\begin{center}
\line(1,0){100}
\end{center}
\begin{footnotesize}
\begin{alltt}
{\bslash}n (shTab,shSum) -> ((),(n:shTab,shSum+n))  
\end{alltt}
\end{footnotesize}
Indeed, for each lambda expressions, the arguments are \texttt{n}, which is the only free variable in the transactions, and \texttt{(shTab,shSum)}, which is a couple representing, intuitively, the value of the \emph{TVar}s at the beginning of the transaction. The lambda expressions express how \emph{TVar}s are updated with respect to those arguments.   
We can also transform the transactional invariant into a (pure) function contract
%in accordance to \cite{static-contract-checking}, 
which bears no reference to STM:

\begin{footnotesize}
\begin{verbatim}
INVARIANT :: Ok -> c -> (Any,c)     where c = {(shTab,shSum) | inv(shTab,shSum)} 
\end{verbatim}
\end{footnotesize}
Intuitively, an expression $e$ satisfies this contract if, when two arguments are applied to $e$ such that the first one does not crash, i.e. hence the \texttt{Ok} that will be defined in Section \ref{details}, and the second one satisfies $c$, then it produces a couple of elements such that the second one also satisfies $c$.  
Verifying whether this contract is satisfied by the three transformed transactions is then an instance of standard contract checking for pure expressions with respect to pure contracts \cite{static-contract-checking}. The verification will prove failure for the two first transactions and success for the third one. 
These results are easily transposable to the original concurrent functions as they state that the function at the top part does not preserve the consistency of the \emph{TVar}s (indicating an application-level race conditions) whereas the one at the bottom part does.

\section{Our Framework in Detail}
\label{details}
In this section, we present a core language based on STM Haskell, denoted $\mathcal{H}$, the language of {\em contracts}, and we formally develop the $\pt$-operator that allows to transform expressions and contracts such that they can be checked by standard contract checking.

\subsection{The language}
\label{language}

A STM Haskell program can be seen as a series of I/O operations. 
Among the different kinds of I/O operations (reading/writing a file, creating a thread,...) is the \texttt{atomically} operation which allows to perform a series of STM operations in a conceptually atomic way with respect to the other threads. Such a series of STM operations embedded in an \texttt{atomically} operation is called a \emph{transaction}. 
STM operations consist basically in reading and updating transactional variables (\emph{TVar}s), and the \emph{only} way to perform these is from within an \texttt{atomically} operation (a fact that is guaranteed by the type system of STM Haskell). Performing I/O operations is \emph{not} allowed inside transactions.

Fig.~\ref{syntax-expression} presents the syntax of $\mathcal{H}$, the core language used for writing a STM Haskell transaction. As such, the language is only a subset of the one defined in~\cite{composable-memory-transactions} and~\cite{compositional-theory-stm} but it allows to define all parts of a STM Haskell program that can be used from within a transaction, including STM operations, lambda abstractions and (recursive) function definitions. Note that $\mathcal{H}$ does not contain an \texttt{atomically} primitive, as the latter is an IO operation and IO operations are not permitted within a transaction. Consequently, transactions cannot be nested in STM Haskell.
For the sake of clarity, we consider that the set of function symbols ($\mathcal{F}$), lambda variables ($\mathcal{X}$), \emph{TVar}s ($\mathcal{V}$) and data constructors ($\mathcal{K}$) underlying a program are finite sets that are pairwise distinct. 
Moreover, we suppose that $\mathcal{V} = \{ t_1, ... ,t_n \}$ is a totally ordered set. We also consider the existence of a mapping $\Delta$ from function names of $\mathcal{F}$ to expressions. In other words, $\Delta$ contains those function definitions that can be called from within a transaction.

\begin{figure*}[htb]
  \centering
%  \fbox{
\begin{math}
    \begin{array}{lcllrrr}
      %\Delta  & \in & \textbf{Programs} \\ 
      %\Delta  & ::= & f = e} \smallskip \\

      % def & ::= & f \in c                  & \mbox{contract attribution} \\
      %     & |   & f %\overline{x} \ 
      %                 = e & \mbox{top-level definition} \smallskip \\

      f   & \in & \textbf{Function Names} \ (\mathcal{F}) \ \smallskip \\

      x,y & \in & \textbf{Lambda Variables} \ (\mathcal{X}) \smallskip \\ 

      K   & \in & \textbf{Data Constructors} \ (\mathcal{K}) \smallskip \\

      t   & \in & \textbf{Transactional Variables} \ (\mathcal{V}) \smallskip \\

      e,p & \in & \textbf{Exp}              & \mbox{\textbf{Expressions}} \\ 
      e,p & ::= & x                        & \mbox{variable}                &  & & [E1] \\
          &  |  & K \ \overline{e}         & \mbox{constructor}             &  &\mbox{(value)} & [E2]\\
          &  |  & \lambda x.e              & \mbox{lambda abstraction}      & &\mbox{(value)} & [E3]\\
          &  |  & r                        & \mbox{exception}               &  &\mbox{(value)} & [E4]\\
          &  |  & e_1 \ e_2                & \mbox{application}             &  && [E5]\\
          &  |  & f                        & \mbox{function call}           & && [E6]\\
          &  |  & \caseof{e}{\{alt_1...alt_n\}} & \mbox{case-expression}    &  && [E7]\\ 
%          &  |  & \code{NoInline} \ e      & \mbox{special tag for contract checking\ } & \mbox{pure expression} \\
          &  |  & \readTVar{t}             & \mbox{STM read variable}       &  && [E8]\\
          &  |  & \writeTVar{t}{e}         & \mbox{STM write variable}      &  && [E9]\\
          &  |  & e_1 \bindTM e_2          & \mbox{STM bind}                    &  && [E10]\\
          &  |  & \returnTM{e}             & \mbox{STM return}                  &  &\mbox{(value)}& [E11]\smallskip \\ 
          % &  |  & e_1 \bindIO e_2          & \mbox{bind}                    & \mbox{IO operation} \\
          % &  |  & \returnIO{e}             & \mbox{return}                  & \mbox{IO operation}\\
          % &  |  & \atomically{e}           & \mbox{transaction execution } & \mbox{IO operation} \\
          % &  |  & \forkIO{e}               & \mbox{thread creation}         & \mbox{IO operation} \\
          % &  |  & e_1 \ | \ e_2            & \mbox{thread composition}      & \mbox{IO operation} \smallskip \\

      r   & \in & \textbf{Exceptions} \\
      r   & ::= & \code{BAD} \\ 
          &  |  & \code{UNR} \smallskip \\

      alt & ::= & K \ \overline{x} -> e \smallskip \\

%      pat & ::= & K \ \overline{x} \\
%          &  |  & \code{DEFAULT} \\

    \end{array}
  \end{math}%}
%\nocaptionrule 
\caption{Syntax of $\mathcal{H}$ expressions}
\label{syntax-expression}
\end{figure*}

A first kind of expressions, which we call \emph{pure expressions} are those constituted by repeated application of only the rules $E1$ - $E7$. They correspond to the language defined in \cite{static-contract-checking}, being a classical functional language based on construction ($E2$), lambda abstraction ($E3$), application ($E5$), function calls ($E6$), and case expression ($E7$). Note the presence of \emph{exceptions} ($E4$) -- basically being the zero-arity predefined constructors \texttt{BAD} and \texttt{UNR}.
A second kind of expressions, which we will call \emph{STM expressions} are those that involve at least one application of a rule among $E8$ - $E11$. 
Retrieval of the content of a TVar ($E8$), updating the content of a TVar ($E9$), binding two STM expressions ($E10$) and defining the return expression ($E11$) are the main operations we can find in a STM expression.

A type system exists for 
%pure expressions \cite{static-conctract-checking-thesis,travauxdebasessurlestypesenhaskell} as for 
STM expressions \cite{compositional-theory-stm} and we will consequently suppose dealing only with well-typed expressions. 
For any expression $e$ we will denote by $e$\texttt{::a} the fact that the expression is of type \texttt{a}. 
In particular, a STM expression having an outermost redex of the form $E8$, $E9$, $E10$ or $E11$ is of type \texttt{STM a} where \texttt{a} is the type of the expression returned when evaluating the STM expression, and a TVar as being of type \texttt{TVar a} where \texttt{a} is the type of the expression that we can store in this TVar. We will refer to STM expressions of type \texttt{STM a} as \emph{STM operations}. 
We will also suppose implicitly that all considered expressions are closed, i.e. there is no free variable, and \emph{well-formed}, that is to say that wherever $\writeTVar{t}{e}$ or $\returnTM{e}$ appear in an expression, $e$ is a pure expression.

In what follows, we suppose certain types and constructors given. Among them the type \texttt{Bool} defining the zero-arity constructors \texttt{True} and \texttt{False}, i.e.  $\{\mathtt{True},\mathtt{False}\}\in \mathcal{K}$. In the examples, we will furthermore use integers and lists, the latter being defined using two constructors: the zero-arity constructor \texttt{[]} to represent empty list and the binary constructor \texttt{(:)} which is often used in an infix way, i.e. \texttt{x:xs} is the list obtained by adding \texttt{x} in front of the list \texttt{xs}. Finally, to enhance readability of the examples, we will sometimes use the convenient so-called \texttt{do}-notation \cite{composable-memory-transactions} and \texttt{let}-notation as a syntactic sugar:

\begin{center}
\begin{math}
  \begin{array}{rcl}
    \mathtt{let } \ x \ \mathtt{ = } \ e' \ \mathtt{ in } \ e \  & \equiv & (\lambda x . e) \ e' \smallskip \\ 
    \mathtt{do \{ }x \codeinmath{<-} e\mathtt{;} S \mathtt{\}} & \equiv & e \ \bind \ (\lambda x . \mathtt{do \{}S\mathtt{\}}) \\
     \mathtt{do \{ } e \mathtt{;} S \mathtt{\}} & \equiv & e \ \bind \ (\lambda \_ . \mathtt{do \{}S\mathtt{\}}) \\ 
    \mathtt{do \{ }  e \mathtt{\}} & \equiv & e \\ 
\end{array}
\end{math}
\end{center}

\newcommand{\cpl}[1]{\ensuremath{\langle #1 \rangle}}

\begin{figure*}[htb]
  \centering
\fbox{\begin{minipage}{0.7\textwidth}
  $\langle (\lambda x.e_1) \ e_2 , \sigma \rangle \rightarrow \langle e_1[e_2/x] , \sigma \rangle \ \ \ (APP) \ \ \ \ \ \ $
%r \ e_2 \rightarrow r \ \ \ (APPR) \ \ \ \ \ \ 
  $\cfrac{f = e \in pgm}{\langle f, \sigma \rangle \rightarrow \langle e, \sigma \rangle} \ \ \ (CALL) $ \medskip \\ 
  $\langle \caseof{K_i \ \overline{e_i}}{\{... , K_i \ \overline{x_i} -> e, ...\}} , \sigma \rangle 
  \rightarrow \langle e[e_i/x_i], \sigma \rangle \ \ \ (CASE) \ \ \ \ \ \ $ \medskip \\
  $\langle \readTVar{t} , \sigma \rangle \rightarrow \langle \returnTM{\sigma(t)} , \sigma \rangle \ \ \ (READ) \ \ \ \ \ \ $ \medskip \\ 
  $\langle \writeTVar{t}{e} , \sigma \rangle \rightarrow \langle \returnTM{\mathtt{()}} , \sigma[ t \mapsto e ] \rangle \ \ \ (WRITE) \ \ \ \ \ \ $ \medskip \\
  $\langle \returnTM e_1 \bindTM e_2 , \sigma \rangle \rightarrow \langle e_2 \ e_1 , \sigma \rangle\ \ \ (BIND) \ \ \ \ \ \ $ \medskip \\
%   $\cfrac{\langle e , \sigma \rangle \rightarrow^{*} \langle \returnTM{e'} , \sigma' \rangle }{
% \langle \atomically{e} , \sigma \rangle \rightarrow \langle \returnIO{e'} , \sigma' \rangle} \ \ \ (ATO) \ \ \ \ \ \ $ \medskip \\
%   $\langle \evalcont[\forkIO{e}] , \sigma \rangle \rightarrow \langle \evalcont[\returnIO{\mathtt{()}}] \ | \ e , \sigma \rangle\ \ \ (FORK) \ \ \ \ \ \ $ \medskip \\
  $\cfrac{\langle e , \sigma \rangle \rightarrow \langle e' , \sigma' \rangle}{\langle \evalcont[\circ / e] , \sigma \rangle \rightarrow \langle \evalcont[\circ / e'] , \sigma' \rangle} \ \ \ (CTX) \ \ \ \ \ \ $
  $ \langle \evalcont[\circ / r] , \sigma \rangle \rightarrow \langle r , \sigma \rangle \ \ \ (EXC) \ \ \ \ \ \ $ 
 \medskip \\

\medskip

$\evalcont ::= \circ \ | \ \evalcont \ e_2 \ | \ \caseof{\evalcont}{\{alt_1...alt_n\}} \ | \  \evalcont \ \bindTM \ e_2 \ $
%| \ (\evalcont | e_2 ) \ | \ (e_1 | \evalcont) 
\end{minipage}}
  %\nocaptionrule
  \caption{Semantics of $\mathcal{H}$ expressions.}
  \label{fig:semantics}
\end{figure*}

An expression is evaluated (or reduced) with respect to an \emph{environment} $\sigma$, which is a mapping from \emph{TVar}s $t_i \in \mathcal{V}$ to pure expressions. We will denote by $\sigma(t_i)$ the (pure) expression that is associated to the \emph{TVar}s $t_i$ in the environment $\sigma$. We will sometimes call this expression the \emph{(transactional) value of $t_i$}. We will denote by $\sigma[ t_i \mapsto e ]$ the environment $\sigma'$ such that $\sigma'(t_i)=e$ and $\forall t_k \in V/\{t_i\} : \sigma(t_k) = \sigma'(t_k)$. The rules by which an expression can be reduced are given in Fig.~\ref{fig:semantics}, in the form of a reduction relation~\cite{principles-program-analysis} $\cpl{e,\sigma} \rightarrow \cpl{e',\sigma'}$ which, from the combination of an expression $e$ and an environment $\sigma$ returns a new expression $e'$ and a new environment $\sigma'$. 
The reduction rules CALL and APP, where we denote the capture-avoiding substitution of $e'$ for each free occurrence of $x$ in $e$ by $e[x/e']$, as well as CASE and BIND are standard rules for functional languages. Note that these reductions have no side-effect in the sense that $\sigma$ is not modified. The rules READ and WRITE define the semantics of reading, respectively updating, a \emph{TVar}. Note that each of these operations reduces to a return operation, allowing to bind their result with a second STM expression (BIND). The rule CTX allows a reduction to be processed in any \emph{context} of the form $\evalcont$ and the rule EXC allows to propagate an exception $r$, which, by the way, can be of any type following the context.
As usual, we denote by $\rightarrow^{*}$ the  reflexive-transitive closure of $\rightarrow$. 
To ease notation, when dealing with expressions other than STM operations, i.e. other than of type \texttt{STM a}, we will omit the environment, i.e. we write $e \rightarrow^{*} e'$ instead of $\cpl{e,\sigma} \rightarrow^{*} \cpl{e',\sigma}$, as the environment is never used nor modified when reducing a such expression. In this case, our semantics coincides with the semantics of the language defined in \cite{static-contract-checking}.   

\newcommand{\diverges}[1]{\cpl{#1}\uparrow^{*}}
\newcommand{\divergespure}[1]{#1\uparrow^{*}}

The attentive reader will notice that 
%we only consider transactions 
%\wim{programs?} 
%with respect to a given and immutable set of TVars. Moreover, 
\emph{TVar}s are considered as global variables: they can be accessed from everywhere in the program, but only through a direct reference. %Moreover, they can not be created dynamically
We will discuss the relevance of these and other limitations of $\mathcal{H}$ in Section~\ref{discussion}.

In what follows, we suppose that expressions to be analyzed are processed beforehand in such a way that missing branches in a case expression are explicitly associated with a \texttt{BAD} exception. For example, if $\mathcal{K} = \{ \mathtt{True}, \mathtt{False} \}$, the expression $\lambda x. \caseof{x}{\{\mathtt{True} -> f\}}$ would be replaced by the expression $\lambda x. \caseof{x}{\{\mathtt{True} -> f, \mathtt{False} -> \mathtt{BAD}\}}$. As in~\cite{static-contract-checking}, we say that an expression crashes if it reduces to \texttt{BAD}.

\begin{definition}
  Let $e$ be an expression and $\sigma$ an environment, $e$ \emph{crashes} in $\sigma$ iff $\cpl{e,\sigma} \rightarrow^{*} \cpl{\mathtt{BAD},\sigma'}$.
\end{definition}

In the particular case of pure expressions, 
we call an expression \emph{crashfree} if and only if there is no way to make it crash (due to a missing pattern) \cite{static-contract-checking}. More formally:
\begin{definition}
  Let $e$ be a pure expression, $e$ is \emph{crashfree} iff $\evalcont[\circ / e] \not\rightarrow^{*} \mathtt{BAD}$ for any context $\evalcont$ such that $\mathtt{BAD}$ does not appear syntactically in $\evalcont$.
\end{definition}

In a similar vein, we say that an expression \textit{diverges} if it cannot be reduced to a value, i.e. a lambda abstraction, a construction, a \texttt{return} or an exception, or if it reduces to \texttt{UNR}. The latter condition will turn out to be interesting in the context of the verification process.

\begin{definition}
  Let $e$ be an expression and $\sigma$ an environment, $e$ \emph{diverges} in $\sigma$, written $\diverges{e,\sigma}$, iff $\cpl{e,\sigma} \rightarrow^{*} \cpl{\mathtt{UNR},\sigma'}  $ % and, later, retry !
 or there is no value $val$ such that $\cpl{e,\sigma} \rightarrow^{*} \cpl{val,\sigma'}$.
\end{definition}

Again, in case of an expression of a type other than \texttt{STM a}, we will often omit the environment from the notation and simply write $\divergespure{e}$ to denote that $e$ diverges.

\subsection{Contracts}
\label{contracts}

\newcommand{\contractTM}[3]{\ensuremath{\parallel #1 \ \lozenge \ #2 \parallel} #3}

The syntax of contracts is given in Fig.~\ref{fig:contractsyntax}. Contracts defined by application of only the rules $C1$ - $C4$ are reserved for specifying contracts on pure expressions and are identical to those defined in \cite{static-contract-checking}. We call them \emph{pure contracts} in order to distinguish them from \emph{STM contracts} which are contracts involving at least one application of the novel rule $C5$. Intuitively, we will associate pure contracts to pure expressions and STM contracts to STM expressions. 

\begin{figure}[htb]
  \centering
%  \fbox{
\begin{math}
    \begin{array}{llllr}
      c & \in & \mbox{\textbf{Contracts}} \\
      c & ::= & \{ x \ | \ p \}        & \mbox{\footnotesize{Predicate Contract}} & [C1] \\ % fv(p) \subseteq \{ x \}
        &  |  & x:c_1 \rightarrow c_2  & \mbox{\footnotesize{Dependent Function Contract}} & [C2]\\
        &  |  & (c_1,c_2)              & \mbox{\footnotesize{Data Constructor Contract}} & [C3]\\
        &  |  & \code{Any}             & \mbox{\footnotesize{Polymorphic Any Contract}} & [C4]\\
        &  |  & \contractTM{x : c_1}{c_2}{c}
                                       & \mbox{\footnotesize{STM Operation Contract}} & [C5]\\ % fv(p_1) \subseteq \{ t \} et fv(p_2) \subseteq \{ t , t' \}, fv(c) \subseteq  \{ t \}
    \end{array}
  \end{math}%}
%\nocaptionrule
  \caption{Syntax of contracts}
  \label{fig:contractsyntax}
\end{figure}

 %\cite{static-contract-checking-thesis}, which for an STM contract of the form $\contractTM{x : c_1}{c_2}{c}$ amounts to specifying that 

We choose this syntax for the contract $C5$, which is called an \emph{STM operation contract}, to fit with the type of STM operations, i.e. \texttt{STM a}. As we will see, the first part -- $\contractTM{x : c_1}{c_2}{}$ -- is related to the (software) transactional memory \texttt{STM} and the second -- the contract $c$ -- to the expression returned, of type \texttt{a}. 
Like expressions, contracts are assumed to be well-typed. 
For the contract $C5$, we expect $c$ to be the kind of contract which is typically associated to expressions of type \texttt{a}. Regarding $c_1$ and $c_2$, they should be contracts for expressions of type $(\mathtt{a_1},...,\mathtt{a_n})$ where $\mathtt{a}_i$ is the type of the expression stored in the \emph{TVar}s $t_i$. 
This idea is expressed more formally by the typing rule in Fig.~\ref{fig:typing}, which extends the typing system for contracts defined in~\cite{static-contract-checking-thesis}.
Note that this implies that $c$, $c_1$ and $c_2$ are required to be pure contracts, as \emph{TVar}s and returned expressions must be pure expressions.

\begin{figure}[htb]
  \centering
  $\cfrac{\forall i : 1 \leq i \leq n : t_i :: \mathtt{TVar \ a_i} ; c_1,c_2 :: (\mathtt{a_1},...,\mathtt{a_n}) ; c::\mathtt{a}}{\contractTM{x : c_1}{c_2}{c} :: \texttt{STM a}}$
  \caption{Typing rule for the STM operation contract.}
  \label{fig:typing}
\end{figure}

The semantics of an expression $e$ \emph{satisfying} a contract $c$, denoted by $e \in c$ is defined in Fig.~\ref{fig:contractsemantics}. The rules $CS1$ - $CS4$ are based on the original work from \cite{static-contract-checking}. 
Intuitively, $e \in \{ x \ | \ p \}$, where $p$ is typically a boolean expression, if $e$ is a sane expression (there is no proper way to make it crash) and the predicate $p[e/x]$ returns \texttt{True}. 
Note that the frequently used contract \texttt{Ok} is just a syntactic notation for a contract $\{ x \ | \ \mathtt{True} \}$.
%
%Let us note that $p$ can be any expression, so it can itself also diverge. 
%
An expression $e$ satisfies $x : c_1 -> c_2$ if it satisfies $c_2$ when given an argument that satisfies $c_1$. Note the use of $x$ which allows to refer from within $e_2$ to the value of the argument.
%DELETE
%We use \emph{dependent} contract, i.e. $x : c_1 -> c_2$, in order to allow $c_2$ to involve the value of the argument in question. 
%
Likewise, an expression $e$ satisfies a pair of contracts if it evaluates to a pair and if each element satisfies its corresponding contract. As a pair is simply a particular constructor from $\mathcal{K}$ with a somewhat nonstandard notation, this contract can effectively be generalized to any constructor from $\mathcal{K}$. 
%, i.e. $(e_1,e_2) \equiv \mathtt{P} \ e_1 \ e_2$ where $\mathtt{P} \in \mathcal{K}$, 
%
The special contract \texttt{Any} is satisfied by any pure expression, including crashing expressions such as \texttt{BAD}. 

\begin{figure*}[htb]
  \centering
  \fbox{\begin{math}
    \begin{array}{lllr}
      e \in \{ x \ | \ p \}     & \Leftrightarrow & e \mbox{ is pure and } ( \divergespure{e} \mbox{or } 
                                                    ( e \mbox{ is crashfree and } \\ && 
                                                  (\divergespure{p[e/x]} \mbox{ or } p[e/x] \rightarrow^{*} \mathtt{True}))) & [CS1] \smallskip \\
      e \in x : c_1 -> c_2      & \Leftrightarrow & \divergespure{e} \mbox{or } 
                                                    ( e \rightarrow^{*} \lambda x.e_2 \mbox{ and } \\ 
                                                 && \forall e_1 \in c_1 : (e \ e_1) \in c_2[e_1/x]) & [CS2] \smallskip\\
      e \in (c_1,c_2)           & \Leftrightarrow & \divergespure{e} \mbox{or } 
                                                    ( e \rightarrow^{*} (e_1,e_2) \mbox{ and }  \\ 
                                                 && e_1 \in c_1,e_2 \in c_2) & [CS3] \smallskip\\
      e \in \mathtt{Any}        & \Leftrightarrow & e \mbox{ is pure.} & [CS4] \smallskip \\
      e \in \ \contractTM{x : c_1}{c_2}{c}
                                & \Leftrightarrow & \forall \sigma : \vv{\sigma} \in c_1 : \diverges{e,\sigma} \\  && \mbox{or } 
                                                    ( \cpl{e,\sigma} \rightarrow^{*} \cpl{\returnTM{e'},\sigma'} \mbox{ and }  \\ 
                                                 && \vv{\sigma'} \in c_2[\vv{\sigma}/x] \mbox{ and } e' \in c[\vv{\sigma}/x]) & [CS5] \smallskip\\

    \end{array}
  \end{math}}
  %\nocaptionrule
  \caption{Contract Satisfaction}
  \label{fig:contractsemantics}
\end{figure*}

The definition $CS5$ is more particular, as it is dedicated to STM operations. 
Intuitively, a STM operation $e$ satisfies $\contractTM{x : c_1}{c_2}{c}$ if, when it is performed with respect to an environment that satisfies $c_1$, it produces a new environment that satisfies $c_2$ and, moreover, it returns an expression satisfying $c$.
In its definition, we use $\vv{\sigma}$ to refer to the pure expression $(e_1,...,e_n)$ where $e_i = \sigma(t_i)$ for all $t_i \in \mathcal{V}$.
Note also the use of $x$ that allows to refer to the input environment both in $c$ and $c_2$.

The attentive reader will notice that the satisfaction of a contract by an expression does not depend on specific requirements about the value of the \emph{TVar}s, which is desirable as we target a static analysis. This is why we ensure that expressions will behave properly with respect to a given contract regardless the environment, which is incidentally omitted in $CS1$ - $CS4$. 
Also note that a diverging expression satisfies any contract and that the expression $p$ in a contract of the form $\{ x \ | \ p \}$ can also diverge. This means that our framework only concerns \emph{partial correctness} and should be paired with a termination analysis \cite{binary-reachability-analysis,automated-termination-analysis} in order to obtain results concerning total correction.

In Fig.~\ref{fig:ex:send}, we show an example of a function with a contract such that it is satisfied by the function definition. 
Note that, as for the others examples, we use here a convenient Haskell-like syntax, which is also the one we use in our prototype, instead of our more formal but equivalent syntax.  
Intuitively, the point of this function is to store a message into a box -- which is represented by the TVar \texttt{box}, to log this action by incrementing a counter -- another TVar \texttt{ct}, and to return the message that was previously stored in the box. This operation is allowed only if one are connected -- i.e. if the expression stored in the TVar \texttt{c} reduces to \texttt{True}. As such, the first part of the contract depicts the fact that there is no specific requirement about the message, apart from being crashfree. In the second part, we require the TVar \texttt{c} to contain \texttt{True}, and we say that, if this requirement holds, the transaction modifies the \emph{TVar}s in a such way that \texttt{c} is still \texttt{True} (we are still connected), the counter \texttt{ct} has been increased, and the expression returned correspond the initial expression stored in \texttt{box}. Note that, while this should not disturb the reader, we have used some syntactic sugar in order to allow tuples of variable in the left part of predicate contracts, i.e. $\{ (x_1,...,x_n) \ | \ p \}$ as a shorthand for $\{ x \ | \ \caseof{x}{(x_1,...x_n) -> p} \}$, and to avoid repeating variable(s) in (dependent) contracts when it is not useful, i.e. $\contractTM{\{x \ |\ p\}}{c_2}{c}$ as a shorthand for $\contractTM{x:\{x \ | \ p\}}{c_2}{c}$. 

\begin{figure}[htb]
  \centering
\begin{footnotesize}
\begin{quote}
\begin{verbatim}
send :: Msg -> STM Msg
send :: Ok  -> || {(c, box, ct ) | c} <> {(c',box',ct') | c' && ct'>ct} || {res|res==box} 
send msg = do connected <- readTVar c
              case connected of 
                True  -> do oldMsg <- readTVar box
                            writeTVar box msg
                            x <- readTVar ct
                            writeTVar ct (x+1)
                            return oldMsg
                False -> BAD
\end{verbatim}
\end{quote}
\end{footnotesize}  
  \caption{The function \texttt{send} and its type and contract.}
  \label{fig:ex:send}
\end{figure}

Note that our extensions to the contract system of \cite{static-contract-checking} are such that the set of desirable properties, discussed in the latter work, still hold. For instance, we can still deal with function calls and recursion in a modular way: to check that the definition of a function $f$ satisfies the contract $c_f$ where $f$ is defined as $f = e$ and $e$ is an expression containing a call to function $g$, we can simply check that $\lambda g . e \in c_g \rightarrow c_f$ where $c_g$ is the contract of $g$. 

Now that we have introduced the basic formalism allowing to specify contracts over expressions, we can now formally define what it means for a contract to be a \emph{transactional invariant}:
\begin{definition}
  An STM operation $e$ is \emph{consistent} with respect to a contract $c$ iff $e \in \ \contractTM{c}{c}{\mathtt{Any}}$. The contract $c$ is then called an \emph{transactional invariant} of $e$.
\end{definition}
For example, the contract $\{ t \ | \ t > 0 \}$  is a transactional invariant of the STM operation 
\[ e = \readTVar t \bindTM \lambda x. \writeTVar{t}{(x+1)} \]

To conclude this section, let us briefly restate the basic idea behind the verification process: first, the programmer writes a contract $c$ that corresponds to his or her view of the consistency over the \emph{TVar}s, and, possibly, contracts for the functions defined by either pure or STM expressions. Secondly, 
it needs to be verified whether the contract $c$ is effectively a transactional invariant for every transactions appearing in the program, a transaction being the STM operation $e$ where \texttt{atomically(}$e$\texttt{)} appears in the program. 
Every transaction for which $c$ cannot be proven to be a transactional invariant possibly represents an application-level race condition.
Before formalizing the verification step, note that our definitions imply that a transaction be a closed expression, i.e. it may not contain free variables. If there \emph{are} free variables in a transaction $e$, rather than verifying $e \in \ \contractTM{c}{c}{\mathtt{Any}}$, one must verify whether  $(\lambda x_1 ...  \lambda x_n . e)\: \in\: c_1 -> \cdots -> c_n -> \ \contractTM{c}{c}{\mathtt{Any}}$,
 where $x_1,...,x_n$ are the free variables in $e$ and $c_1,...,c_n$  the contracts associated to $x_1,...,x_n$. 
This is what was already illustrated by the example at the end of Section~\ref{idea}, where the contract \texttt{Ok} was associated to the (initially) free variable \texttt{n}. Note that, in practice, these contracts can be provided by the programmer indirectly, through a top-level contract attribution, or -- to some extent -- be automatically generated with by-default value \cite{static-contract-checking}. This, however, is out of the scope of the current work.

\subsection{Checking through Program and Contract Transformation}
\label{checking}

Our approach in checking whether a program written in STM Haskell is free of application-level race conditions consists then in transforming a transaction $e$ and its contract $c$ into a pure expression $e'$ and a pure contract $c'$, in a such way that $e' \in c'$ implies $e \in c$. This transformation, represented by the $\pt$-operator, is defined -- for expressions -- in Fig.~\ref{fig:pt}. For sake of clarity, the transformation assumes that only a single TVar is handled by the program, i.e. $\mathcal{V} = \{t\}$, but the definitions can readily be extended towards handling a given set of \emph{TVar}s.
\begin{figure*}[t]
  \centering
%  \fbox{
\begin{math}
    \begin{array}{lllllll}
      %\multirow{7}{*}{\code{STM} \ $\left\{\rule{0mm}{15mm}\right.$} 
     & \pt(\readTVar{t})     & = & \lambda t . (t,t) \\
     & \pt(\writeTVar{t}{e}) & = & \lambda t . (\code{()},e) \\
     & \pt(e_1 \bindTM e_2)  & = & \lambda t . (\pt(e_2) \ (\code{fst} \ e_1') \ (\code{snd} \ e_1')) & \mbox{where} \ & e_1'  & = \pt(e_1) \ t \\
      &                       &   &                                                              &                & \code{fst}& = \lambda (a,b) . a \\
     &                       &   &                                                              &                & \code{snd}& = \lambda (a,b) . b \\
     & \pt(\returnTM{e})     & = & \lambda t . (e,t) \\  
     %\multirow{4}{*}{\code{ IO} \ $\left\{\rule{0mm}{8mm}\right.$} 
     % & \pt(e_1 \bindIO e_2)  & = & \pt(e_2) \ \pt(e_1) \\
     % & \pt(\returnIO{e})     & = & \pt(e) \\
     % & \pt(\atomically{e})   & = & (\code{NoInline} \ \atomicallypONLY) \ \pt(e) \\
     % & \pt(\forkIO{e})       & = & (\code{NoInline} \ \forkIOpONLY) \ \pt(e) \\ 
     & \pt(f)                & = & \lambda t . (f \ t)                                            & \multicolumn{3}{l}{\mbox{if $f :: \code{STM a}$.}}\\
     &                       & = & f                                                            & \multicolumn{3}{l}{\mbox{else.}}\\
     & \pt(e_1 \ e_2)        & = & \lambda t . (\pt(e_1) \ \pt(e_2) \ t )                        & \multicolumn{3}{l}{\mbox{if $e_1 \ e_2:: \code{STM a}$.}}\\
     &                       & = & \pt(e_1) \ \pt(e_2)                                                    & \multicolumn{3}{l}{\mbox{else.}}\\%
     & \pt(\lambda x . e)    & = & \lambda x . \pt(e) \\
     & \pt(K \ \overline{e}) & = & K \ \overline{\pt(e)} \\
     & \pt(x)                & = & x \\
     &  \pt(r)                & = & r \\
     & \pt( e_{\mathtt{case}}) & = & \multicolumn{2}{l}{\lambda t . ((\caseof{\pt(e)}{\overline{pat_i -> \pt(e_i)}}) \ t)} & \multicolumn{2}{l}{\mbox{if  $e_{\mathtt{case}} :: \code{STM a}$.}}\\
     &                       & = & \multicolumn{2}{l}{\caseof{\pt(e)}{\overline{pat_i -> \pt(e_i)}}} & \multicolumn{2}{l}{\mbox{else.}}\\
     &                        \multicolumn{3}{l}{\ \ \ \ \mbox{where } e_{\mathtt{case}} \equiv\caseof{e}{\overline{pat_i -> e_i}} } 
     %& \pt(f)                & = & f \\
     %& \pt(e_1 \ e_2)        & = & \pt(e_1) \ \pt(e_2) \\
     %& \multicolumn{6}{l}{\pt(\caseof{e}{\overline{pat_i -> e_i}})} \\ && = & \multicolumn{2}{l}{\caseof{\pt(e)}{\overline{pat_i -> \pt(e_i)}}} \\
   \end{array}
  \end{math}%}
  %\nocaptionrule
  \caption{$\pt$-transformation for $\mathcal{H}$ expressions}
  \label{fig:pt}
\end{figure*}

The basic idea is to transform a STM operation $e$ returning $e'$ and possibly updating the (transactional) value of $t$ into a lambda expression $\lambda t.(e',e'')$ where $e''$ is the updated value of $t$. 
Reading the TVar is modeled by a function that associates to the initial value of $t$ the couple $(t,t)$. This depicts the fact that the operation returns the value stored in the TVar (the first $t$) while the value stored in $t$ does not change (the second $t$). 
The update operation does not return a relevant expression (so we return the nullary constructor \texttt{()}, as usual in Haskell) but it replaces the expression that was previously stored in $t$ by the given expression. 
Note that, as we consider well-formed expressions, the expression $e$ that will be written into the TVar is pure, and hence need not be transformed.
When a STM operation is binded with another STM expression, both expressions are converted and the return expression and the updated TVar value of the first transformed expression are provided as input for application with the second transformed expression. Note the use of \texttt{fst} and \texttt{snd} that retract, respectively, the first and the second expressions from a pair.
To help understanding the intuition behind this transformation, we can see that 
\begin{center}
  \begin{math}
    \begin{array}{c} %\centering
      \pt(\readTVar{t} \bindTM \lambda x.\writeTVar{t}{(x+1)}) \\
      = \\
      \lambda t.((\lambda x.\lambda t.(\code{()},x+1)) \ (\code{fst} \ ((\lambda t.(t,t)) \ t)) \ (\code{snd} \ ((\lambda t.(t,t)) \ t))
    \end{array}
  \end{math}
\end{center}
which can  \emph{symbolically} \cite{compiling-haskell-program} be rewritten into $\lambda t . (\code{()},t+1)$, clearly reflecting the update of $t$ by $t+1$.
Transforming a return operation is more straightforward as it does not change the value of the \emph{TVar}.
%
%An IO computation is an expression of type \texttt{IO a}. Intuitively, it means that such an expression process to some I/O operation on the environment, and return a value of type \texttt{a}. 
%Unlike \STM \ expression, interferences can disturb a serie of IO computations, in a concurrent context. 
%Performing a transaction, through the call to \atomicallyONLY (which is of type \texttt{STM a -> IO a}) and creating a new thread (which is of type \texttt{IO () -> IO ()}) are two basic IO computations.  
%
%The operator $\pt$ transforms an expression of type \texttt{IO a} into an expression of type \texttt{a}. 
%This makes sense in our context because we are interested in the value of TVars but we have no guarantee about their value once outside of a transaction, i.e. their value can be changed by another thread. 
%As an IO computation \emph{is not} a transaction, while it can \emph{contains} some transactions, we only keep from such expressions the value that is returned. 
%
%The \texttt{NoInline} constructor is a special tag for the contract checking, which is described below, and it has no impact on the execution of expressions.
%
% 
Transforming the remaining expressions basically boils down to propagating the transformation to their subexpressions, as the latter may contain STM operations.
Note that our transformation $\pt$-operator is defined such that:
\begin{itemize}
\item if $e$ is a STM operation, then the execution of $\pt(e) \ e'$ gives a couple $(e_1,e_2)$ where $e_1$ is the expression that would be returned by the STM operation $e$, and $e_2$ is the expression that would be stored finally in the TVar if the latter would have contained $e'$ before performing the STM operation.  
%\item if $e$ is a IO computation, then the execution of $\pt(e)$ gives $\bot$ if the value returned by $e$ has something to do with the TVar; otherwise it gives the expected value 
\item if $e$ is a pure expression, then the execution of $\pt(e)$ will produce the same result as the execution of $e$.
\end{itemize}
Moreover, the $\pt$-operator can be generalized easily for a (totally ordered) set of $n$ \emph{TVar}s, i.e. $\mathcal{V} = \{ t_1,...,t_n \}$. This would require transforming into an expression that takes as argument a tuple $(t_1,...,t_n)$ instead of a single $t$ in Fig~\ref{fig:pt} and that returns a couple where the second element is a new tuple of $n$ expressions. For example, the first rule would look like:\footnote{we use the following syntactic sugar : $ \lambda (t_1,...t_n) . e \equiv \lambda x . \caseof{x}{(t_1,...,t_n) -> e}$ where $x$ is not a free variable of $e$.}
\begin{center}
  \begin{math}
    \begin{array}{lll}
      \pt(\readTVar{t_k})     & = & \lambda (t_1,...,t_k,...,t_n) . (t_k,(t_1,...,t_k,...,t_n)) \\
    \end{array}
  \end{math}
\end{center}

We also need a technique to convert STM contracts into pure contracts. 
For this purpose, we override the $\pt$-operator such that it deals with contracts too. The transformation is rather straightforward and depicted in Fig.~\ref{fig:pt:contract}. The main point is the conversion of a STM operation contract into a dependant function contract in order to fit with the form of a transformed STM operation.

\begin{figure}[htb]
  \centering
%  \fbox{
\begin{math}
    \begin{array}{rll}
       \pt(\{ x \ | \ p \}) & = & \{ x \ | \ p\}\\
        \pt(x:c_1 \rightarrow c_2)  & = & x:\pt(c_1) \rightarrow \pt(c_2) \\
        \pt(c_1,c_2) & = & (\pt(c_1),\pt(c_2))    \\
        \pt(\code{Any}) &=& \code{Any}\\
        \pt(\contractTM{x : c_1}{c_2}{c}) &=& x : c_1 \rightarrow (c,c_2) \\ 
    \end{array}
  \end{math}%}
  %\nocaptionrule
  \caption{$\pt$-transformation for contracts}
  \label{fig:pt:contract}
\end{figure}

Fig.~\ref{fig:pt:prop} enlists a number of easily proven but important properties of the $\pt$-operator. 
Property~\eqref{prop:pure} states that the resulting expression, or contract, is pure in the sense that no STM-related construction remains after the transformation. This follows immediately from the definition of $\pt$ and the fact that we consider only well-formed contracts, i.e. subcontracts of STM contracts are pure. 
Secondly, $\pt$ is an idempotent operator \eqref{prop:idempotence}  and it does not change its argument when the latter is already pure (\ref{prop:pure:equiv:e} - \ref{prop:pure:equiv:c}).
Further, it follows that the transformation of a STM operation different from \texttt{BAD} and \texttt{UNR} always reduces  to a lambda abstraction \eqref{prop:ptconv}.
%The property~\eqref{prop:etcimpliestee:e} depicts the fact that, by definition of $\in$, if an expression satisfies a pure contract, it means necessarily that the expression is also pure. 
%
The two last properties \eqref{prop:div} and \eqref{prop:conv} which can be proved by induction on $e$, are fundamental in our framework and depict the equivalence that links the semantics of a STM operation with its transformed pure counterpart. 

\begin{figure}[htbp]
  \centering
%  \begin{math}
    \begin{equationarray}{llr}
    \bullet & \mbox{$\pt(e)$, resp. $\pt(c)$, is a pure expression, resp. contract.} \label{prop:pure} \\
    \bullet & \pt(\pt(e)) \equiv \pt(e), \pt(\pt(c)) \equiv \pt(c)                             \label{prop:idempotence} \\
    \bullet & \pt(e) \equiv e \mbox{ if $e$ is a pure expression.}                             \label{prop:pure:equiv:e} \\
    \bullet & \pt(c) \equiv c \mbox{ if $c$ is a pure contract.}                             \label{prop:pure:equiv:c} \\
    \bullet & \pt(e) \rightarrow^{*} \lambda x.e' \mbox{ if $e \ \mathtt{:: STM \ a}$ and $e \not\equiv r$.} \label{prop:ptconv} \\
%    \bullet & e \in \pt(c) \implies \pt(e) \equiv e                              \label{prop:etcimpliestee:e} \\
%    \bullet & \pt(e) \in c \implies \pt(c) \equiv c                              \label{prop:etcimpliestee:c} \\  
    \bullet & \diverges{e,\sigma} \iff \divergespure{\pt(e) \ \vv{\sigma}} \mbox{ for $e \ \mathtt{:: STM \ a}$}      \label{prop:div} \\  
    \bullet & \cpl{e,\sigma_1} \rightarrow^{*} \cpl{\returnTM{e'},\sigma_2} \iff \pt(e) \ \vv{\sigma_1} \rightarrow^{*} (e',\vv{\sigma_2}) \ \\ 
            & \mbox{ for $e \ \mathtt{:: STM \ a}$} \label{prop:conv}  
    \end{equationarray}
%  \end{math}
  \caption{Properties of $\pt$}
  \label{fig:pt:prop}
\end{figure}

Trivially, as a direct consequence of properties \eqref{prop:pure:equiv:e} and \eqref{prop:pure:equiv:c}, we have that $e \in c \iff \pt(e) \in \pt(c)$ when $e$ and $c$ are a pure expression and contract. 
More importantly, the same property holds for a STM operation and a STM operation contract of same type, as stated by the following theorem.
\begin{theorem}
Let $e$ be an expression of type $\mathtt{STM \ a}$ and let $c$ be a contract of type $\mathtt{STM \ a}$, 
\[ e \in c \iff \pt(e) \in \pt(c) \]
\end{theorem}
\begin{proof}
As $c$ is of type \texttt{STM a}, it is a contract of the form $c \ \equiv \ \contractTM{c_1}{c_2}{c'}$. We distinguish three cases:
\begin{itemize}
\item $e \equiv \mathtt{BAD}$ : 
We can see, by the semantics of contracts and expressions, that $\mathtt{BAD} \not\in \ \contractTM{c_1}{c_2}{c'}$ and $\pt(\mathtt{BAD}) \equiv \mathtt{BAD} \not\in \ x : c_1 -> (c',c_2) \equiv \pt(\contractTM{c_1}{c_2}{c'}) $ for any $c_1,c_2,c'$. 
\item $e \equiv \mathtt{UNR}$ : 
Similarly, $\mathtt{UNR} \in \ \contractTM{c_1}{c_2}{c'}$ and $\pt(\mathtt{UNR}) \equiv \mathtt{UNR} \in \ x : c_1 -> (c',c_2) \equiv \pt(\contractTM{c_1}{c_2}{c'}) $ for any $c_1,c_2,c'$. 
\item $e \not\equiv \mathtt{BAD},\mathtt{UNR}$ : \medskip\\
  \begin{math}
    \begin{array}{ll}
      & \pt(e) \in \pt(c) \smallskip \\
      \iff & \mbox{(form of $c$)}\smallskip\\
      & \pt(e) \in \pt(\contractTM{c_1}{c_2}{c'})\smallskip \\
      \iff & \mbox{(def. of $\pt$)} \smallskip\\
      & \pt(e) \in x : c_1 -> (c',c_2) \smallskip\\
      \iff & \mbox{(def. of $\in$ + prop. \eqref{prop:ptconv})} \smallskip\\
      & \forall e_1 \in c_1 : (\pt(e) \ e_1) \in (c',c_2)[e_1/x] \smallskip\\ 
      \iff & \mbox{(def. of $\sigma$ + type of $c_1$)} \smallskip\\
      & \forall \sigma : \vv{\sigma} \in c_1 : (\pt(e) \ \vv{\sigma}) \in (c',c_2)[\vv{\sigma}/x]\smallskip \\
      \iff & \mbox{(def. of $\in$ + type of $c_2$)} \smallskip\\
      & \forall \sigma : \vv{\sigma} \in c_1 : \divergespure{\pt(e) \ \vv{\sigma}} \mbox{or } 
                                                    ( (\pt(e) \ \vv{\sigma}) \rightarrow^{*} (e',\vv{\sigma'}) \\ & \mbox{ and } e' \in c'[\vv{\sigma}/x], \vv{\sigma'} \in c_2[\vv{\sigma}/x]) \smallskip\\
      \iff & \mbox{(prop. \eqref{prop:div} and \eqref{prop:conv})} \smallskip\\
      & \forall \sigma : \vv{\sigma} \in c_1 : \diverges{e,\sigma} \mbox{or } 
                                                    ( \cpl{e,\sigma} \rightarrow^{*} \cpl{\returnTM{e'},\sigma'} \\ & \mbox{ and } e' \in c'[\vv{\sigma}/x]), \vv{\sigma'} \in c_2[\vv{\sigma}/x] \\    
      \iff &  \mbox{(def. of $\pt$)} \smallskip\\
      & e \in \ \contractTM{c_1}{c_2}{c'} \smallskip\\                             
      \iff &  \mbox{(form of $c$)}\smallskip\\ 
      & e \in c                              
    \end{array}
  \end{math}
\end{itemize}
\end{proof}

The following corollary generalizes the above result towards any expression and contract (of the same type). It can be easily proven by induction on the structure of the expression and corresponding contract, using the results for a pure subexpression (contract) and a subexpression (contract) of type \texttt{STM a} as base cases.
%Finally, thanks to the result of the above theorem, we can generalize the idea of the equivalence to any expression and contracts of same types. Indeed, the following corollary can be proved in a straightforward way by induction on dependant function contract and data constructor contract, using the result of our theorem for the base case. 
\begin{corollary}
  Let $e$ be an expression of type $\mathtt{a}$ and let $\mathtt{c}$ be a contract of type $a$, 
  \[ e \in c \iff \pt(e) \in \pt(c) \]
\end{corollary}

The above corollary basically states correctness of approach: no precision is lost by transforming a STM Haskell transaction into a pure function and using standard techniques for contract checking \cite{static-contract-checking} to prove its consistence with respect to the transactional invariant. The approach has been fully implemented and the concerned reader is invited to try the prototype of our framework\footnote{Our prototype can be  downloaded at the following URL: \\ \texttt{http://www.info.fundp.ac.be/$\mathtt{\sim}$rde}}.

\subsection{Limitations and extensions}
\label{discussion}

\subsubsection*{Blocking and Composable Transactions} 
STM Haskell also allows to define blocking and composable (alternatives) STM operations by means of the primitives \texttt{retry :: STM a} and \texttt{orElse :: STM a -> STM a -> STM a} \cite{composable-memory-transactions}.
The first one is an outstandingly simple way to force a thread to wait for an event. Semantically, \texttt{retry} just make the transaction abort, i.e. all changes are discarded, and the transaction is restarted from the beginning (in practice, it restarts only when a related TVar is modified, but this is only an implementation detail). For example, in the following call, the transaction will conceptually wait for the TVar \texttt{connected} to be \texttt{True} before continuing to the operation \texttt{sendMessage}.

\begin{footnotesize}
\begin{verbatim}
atomically ( do c <- readTVar connected
                case c of True  -> sendMessage
                          False -> retry )
\end{verbatim}  
\end{footnotesize}
There is no difficulty to deal with \texttt{retry} in our framework. It is sufficient to define $\pt(\mathtt{retry}) = \mathtt{UNR}$. 
Indeed, recall that we do not have to deal with diverging expressions, nor with the values of the \emph{TVar}s at the intermediate states of the transaction, but only with the possible values of the \emph{TVar}s at the very end of the transaction, which are not influenced by a \texttt{retry} branch.

The other primitive, \texttt{orElse}, allows multiple STM operations to be composed as alternatives. Basically, $e_1 \ \mathtt{`orElse`} \ e_2$ proceeds as follows: first, $e_1$ is evaluated. If its evaluation does \emph{not} result in \texttt{retry}, the operation ends with the result computed by $e_1$. If on the other hand, $e_1$ \emph{does} evaluate to \texttt{retry}, rather than restarting the operation, $e_2$ is evaluated. Only if the latter also result in \texttt{retry} is the entire operation restarted.

In order to verify a transaction $e$ containing an \texttt{orElse} operation, it suffices to verify multiple versions of the transaction, say $\Gamma(e)$, one version for each possible combination of alternative evaluations. The $\Gamma$-operator computing all possible such evaluations is partly defined below:
\smallskip

\begin{center}
\begin{math}
\begin{array}{rll}
\Gamma(e_1 \ \mathtt{`orElse`} \ e_2) & =&  \Gamma(e_1) \cup \Gamma(e_2) \smallskip \\
\Gamma(e_1 \bindTM e_2) &=& \{ e_1' \bindTM e_2' \ | \ e_1' \in \Gamma(e_1), e_2' \in \Gamma(e_2) \} \smallskip \\
\Gamma(\lambda x . e) & =&  \{ \lambda x . e' \ | \ e' \in \Gamma(e) \} \smallskip \\
& ... & \\
\end{array}
\end{math}
\end{center}
Instead of verifying $e \in c$, we have then to verify that $\forall e' \in \Gamma(e) : e' \in c$.

\subsubsection*{\emph{TVar}s as arguments of a function} 
One limitation of our language $\mathcal{H}$ is that, contrary to full STM Haskell, it does \emph{not} allow a function (or, more precisely, a lambda abstraction) to have \emph{TVar}s as its arguments.  
However, as long as we consider a fixed set of \emph{TVar}s manipulated by the program, this limitation can be overcome by using well-known techniques from program specialization~\cite{partial-evaluation-automatic,introduction-program-specialisation} in order to specialize both the function and its corresponding contract with respect to all possible subsets of \emph{TVar}s that can be provided as actual arguments in a call to the function.
Let us discuss this idea informally by looking at a simple example. Consider the following function, which takes an argument \texttt{x} supposed to be a TVar containing an integer, and which updates the given TVar by incrementing this value.

\begin{footnotesize}
\begin{verbatim}
f :: TVar Int -> STM ()
f x = do n <- readTVar x
         writeTVar x (n+1)
\end{verbatim}
\end{footnotesize}
Suppose that we want to express in a contract that this function must be called with respect to a TVar containing a positive value, and that the resulting value after the update must be strictly greater than the initial value. This could be expressed by a new kind of contract such as the following in which \texttt{t} and \texttt{t'} are used to refer to the initial, respectively, final value of the function's argument:

\begin{footnotesize}
\begin{verbatim}
f :: TVar[t,t'] -> | t >= 0 <> t' > t | Any 
\end{verbatim}
\end{footnotesize}
Supposing now that the program manipulates two \emph{TVar}s, say $\mathcal{V} = \{ \mathtt{tA} , \mathtt{tB} \}$, we can generate \emph{two} versions of the function and its associated contract: one explicitly referring to $\mathtt{tA}$, the other to $\mathtt{tB}$. Let us take, for example, the result of specializing w.r.t. to $\mathtt{tA}$ (the result for $\mathtt{tB}$ is of course similar):

\newcommand{\tun}{tA}
\newcommand{\tdeux}{tB}
\begin{footnotesize}
\begin{verbatim}
f_tA :: || { (tA,tB) | tA >= 0 } <> { (tA',tB') | tA' >= tA } || Any 
f_tA = do n <- readTVar tA
          writeTVar tA (n+1)
\end{verbatim}
\end{footnotesize}
Then, every call to \texttt{f tA} will be replaced by a call to \texttt{f\_tA}. The specialized function no longer contains a TVar as argument, and hence can be verified (w.r.t. to the specialized contract) by our standard framework.

Note that this solution works as long as the set of manipulated \emph{TVar}s is fixed (there are no dynamically created \emph{TVar}s), and are not aggregated into a (possibly recursive) data structure. Whether the technique can be adapted to these situations is an interesting topic for further research.

\section{Conclusion and Related Work}
\label{conclusion}

%\paragraph{Related Work}

%

%Plenty of works targets the detection of concurrency bugs \cite{comprehensive}.  
Data races being one of the most common sources of concurrency bugs, a lot of static \cite{survey-methods-preventing,effective-static-race,warlock-static-data,racerx-effective-static,kiss-keep-simple,race-checking-context,automated-type-based,locksmith-context-sensitive,type-based-race,ownership-types-safe,programming-model-concurrent} and dynamic \cite{literace-effective-sampling,eraser-dynamic-data,hybrid-dynamic-data,racetrack-efficient-detection,goldilocks-race-transaction,randomized-active-atomicity} analysis techniques have been developed to detect them \cite{optimal-tracing-replay} over the last couple of decades. 
The other categories of concurrency bugs discussed broadly in the literature are atomicity violation \cite{type-effect-system,avio-detecting-atomicity,atom-aid-detecting} and locking-related bugs \cite{eraser-dynamic-data}. 
However, very few attention has been given to more higher-level kind of bugs, such as those involving multiple variables, while studies have pointed out that they represent one third of the non-deadlock concurrency bugs found in real-world programs \cite{learning-mistakes-comprehensive}. 
Moreover, we are convinced that, when using higher-level synchronization techniques/concurrent languages -- like STM Haskell -- which allow to avoid problems related to atomicity or locks, the proportion of those bugs involving multiple variables is greatly increased.

% In particular, different approaches exist to deal with the detection of \emph{higher}-level anomalies in concurrent STM programs. These approaches differ in the distinct range of bugs they can handle. That suggests the existence of different compromises ranging from sometimes too strict to sometimes too weak correctness criteria. 

Several frameworks try to tackle the problem of multiple variables \cite{muvi-automatically-inferring,colorsafe-architectural-support,finding-concurrency-bugs,high-level-data,detection-transactional-memory,practical-verification-high-level}.
In those works, static or dynamic techniques are used in order to infer possible correlations between shared variables. 
Then, they allow to detect races related to a set of correlated variables.
This idea is depicted by the concept of {\em high-level data race}, introduced by \cite{high-level-data}. 
Such a race occurs when a set of shared memory locations is meant to be accessed atomically, but the memory locations are accessed separately somewhere in the program \cite{detection-transactional-memory,practical-verification-high-level}. 
The interest of these tools is indisputable as they find a lot of races without requiring (a lot of) annotations, but they typically can trigger false positives and false negatives. 
Indeed, two shared variables can be accessed together but in a bad way while they can also be accessed separately without necessarily violating the link between them.
In other words, while finding the adequate granularity of the atomic region is necessary, the \emph{effect} of this atomic region on the data also matters.  

To overcome this limitation, we have to deal with the subtly variant concept of {\em application-level race condition}, which is introduced and handled (to the best of our knowledge) exclusively by \cite{verifying-correct-usage}. This notion not only captures the existence of a link between shared variables, but {\em explicitly} involves the {\em nature} of that link. By defining under what condition another thread can violate a program invariant, detection of application-level race conditions is more accurate, avoiding a lot of false positives and negatives. On the other hand, analyzes are typically heavier to use as the programmer has to provide additional annotations.    
Although the goal of our work is similar to the static analysis developed for verifying atomic blocks in an object-oriented language \cite{verifying-correct-usage}, the method to achieve that goal is very different. This is not not only because we need to address distinct language-related issues such as the particular control-flow inherent to higher-order functional programming in our case, or access permissions and unpacking methodologies related to object-oriented programming in \cite{verifying-correct-usage}. 
%or  \cite{verifying-correct-usage} mainly deals with access permissions and unpacking methodologies related to object-oriented programming. 
A more fundamental difference is that, in our framework, the consistency definition is written using \emph{the same language} as the program, while \cite{verifying-correct-usage} relies on a distinct formalism, \emph{typestate} \cite{typestates-for-objects}.
As a result, we feel our approach is both more convenient and expressive as we use full Haskell to express our invariants.

%To the best of our knowledge, our work is also the first attempt to use contract checking in the context of a \emph{concurrent} functional language.
As we have basically shown how the problem of detecting application-level race conditions can be recasted in the setting of contract checking for non-concurrent programs \cite{static-contract-checking,sound-complete-models}, it is to be expected that recent improvements for rendering the verification of such contracts more practical \cite{halo-haskell-logic} will have a positive influence on the number and kind of races that can be detected by our approach.

Finally, let us note that consistency of STM in Haskell can also to some extent be checked dynamically, by using the language primitive \texttt{always} (previously called \texttt{check}) \cite{transactional-memory-data}. For our running example, we could write the following property that would be checked at the end of every transactions performed:

\begin{footnotesize}
\begin{verbatim}
always ( do tab <- readTVar shTab
            s <- readTVar shSum 
            return (sum tab == s) ) 
\end{verbatim}
\end{footnotesize}
While guaranteeing consistency during the program's execution, being a dynamic technique, it cannot be used to statically prove that the program is application-level race condition free. It would be interesting to see to what extent such dynamic verification can be coupled with static checking in order to improve the detection of such race conditions and/or perform a more detailed error reporting.

\paragraph{Acknowledgments.} 
We thank the anonymous reviewers for their constructive comments on a previous version of this paper.

\bibliography{fichier_bib_generique_a_utiliser_bien_classe2} 
\bibliographystyle{eptcs}

%\appendix

\end{document}